\documentclass[12pt,reqno,fleqn]{amsart}

\usepackage{latexsym}
\usepackage{amssymb}
\usepackage{amsxtra}
\usepackage{amsmath}
\usepackage{amsfonts}
\usepackage{CJK}
\newcommand{\cleqn}{\setcounter{equation}{0}}
\newcommand{\clth}{\setcounter{theorem}{0}}
\newcommand {\sectionnew}[1]{\section{#1}\cleqn\clth}

\newtheorem{theorem}{Theorem}[section]

\newtheorem{proposition}[theorem]{Proposition}

\newtheorem{remark}[theorem]{Remark}

\renewcommand{\L}{\mathcal{L}}

\newcommand{\M}{\mathcal{M}}
\newcommand{\Z}{\mathbb{Z}}

\newcommand{\B}{\mathcal{B}}

\def\({\left(}
\def\){\right)}
\def\[{\begin{eqnarray}}
\def\]{\end{eqnarray}}
\def\d{\partial}

\begin{document}

\title{Symmetries and reductions on the noncommutative Kadomtsev-Petviashvili and Gelfand-Dickey hierarchies}
\author{Chuanzhong Li}
\dedicatory {  Department of Mathematics, Ningbo University, Ningbo, 315211 Zhejiang, P.\ R.\ China\\
Email:lichuanzhong@nbu.edu.cn}

\thanks{}

\texttt{}
\date{}

\begin{abstract}
In this paper, we construct the additional flows  of the noncommutative Kadomtsev-Petviashvili(KP) hierarchy and the additional symmetry flows constitute an infinite dimensional Lie algebra $W_{1+\infty}$. In addition, the generating function of the additional symmetries can also be proved to have a nice form in terms of wave functions and this generating symmetry is used to construct the noncommutative KP hierarchy with self-consistent sources and the constrained noncommutative KP hierarchy.  The above results will be further generalized to the noncommutative Gelfand-Dickey hierarchies which contains many interesting noncommutative integrable systems such as the noncommutative KdV hierarchy and noncommutative Boussinesq hierarchy. Meanwhile, we construct two new noncommutative systems including odd noncommutative C type Gelfand-Dickey and even noncommutative C type Gelfand-Dickey hierarchies. Also using the symmetry, we can construct a new noncommutative  Gelfand-Dickey hierarchy with self-consistent sources. Basing on the natural differential Lax operator of the noncommutative  Gelfand-Dickey hierarchy, the string equations of the noncommutative  Gelfand-Dickey hierarchy are also derived.
\end{abstract}

\maketitle
Mathematics Subject Classifications (2010):  37K05, 37K10, 35Q53.\ \ \\
\ \ \ Keywords:  noncommutative  KP hierarchy, noncommutative Gelfand-Dickey hierarchy, noncommutative Gelfand-Dickey hierarchy with self-consistent sources, odd noncommutative C type Gelfand-Dickey hierarchy, even noncommutative C type Gelfand-Dickey hierarchy, additional symmetry, $\emph{W}_{1+\infty}$ Lie algebra, String equation.\\
\allowdisplaybreaks
 \setcounter{section}{0}

\sectionnew{Introduction}

The Kadomtsev-Petviashvili(KP) hierarchy(\cite{A},\cite{B}) is one of the most important integrable hierarchy. It arises in many different fields of mathematics and physics such as enumerative algebraic geometry, topological field theory, string theory and so on. The KP hierarchy has the well-known Virasoro  symmetry which was extensively studied in literature(\cite{B}-\cite{C}).

 The noncommutative field theory is  a fruitful subject
in both mathematics and physics particularly in noncommutative integrable systems \cite{5,6,7,KPTodaNC,JNMPnonToda}. The noncommutative theory gives rise to various new physical objects  in quantum mechanics such as the canonical commutation
relation $[q,p]=i\hbar$.
As said in \cite{Hamanaka}, the noncommutative parameter is
closely related to the  existence of a background flux in the effective theory of D-branes. With the presence of background magnetic fields the
noncommutative gauge theories were found to be equivalent to ordinary
gauge theories and noncommutative solitons play important roles in the study of
D-brane dynamics\cite{Dbrane}.

Similar to the additional symmetry flows of the KP hierarchy which was given by Orlov and Shulman(\cite{C}), the  additional symmetry flows  can constitute a centerless $W_{1+\infty}$  algebra.
 Therefore  motivated by the results on the classical KP hierarchy(\cite{G},\cite{H}) , we will give generating functions of the additional symmetries of the noncommutative  KP hierarchy in terms of wave functions which can lead to the noncommutative KP hierarchy with self-consistent sources and the
constrained noncommutative KP hierarchy.

The Gelfand-Dickey hierarchy was introduced by I. M. Gelfand and L. A. Dickey \cite{GD} and  has attracted a lot of attention in the research of
integrable systems. The  Hamiltonian theory of Gelfand-Dickey hierarchy as developed
 in terms of Lax pair which can be seen from the
Dicky's book\cite{B} for detail.
As we know, the string equation formally as $[P,Q]=1$ which connects the Lax operator and
Orlov-Shulman operator is very useful in application on the partition function of the
 string theory \cite{nakatsu}.  The possible
physical interest
in a non-commutative generalization of Orlov's work \cite{orlovwinternitz,orlovwinternitz2} might be an exciting and interesting subject and this becomes one important motivation for us to consider additional symmetries of noncommutative Gelfand-Dickey hierarchy of which the B\"acklund transformation was constructed in \cite{heJHEP}.

The non-commutative KP and Gelfand-Dickey hierarchies have shown to be
Hamiltonian systems for example in \cite{5,7,2}. Non-commutative
Hamiltonian systems in finite dimension \cite{8} and infinite dimension
\cite{1} have a nice algebraic formulation and can be described for any
associative algebra.
In this paper, we will restrict the associative algebra to the case under the
Moyal product which will reviewed later in detail in the next section.

The organization of this paper is as follows. We firstly review the Lax equation of the noncommutative  KP hierarchy in Section 2. In Section 3, under the basic Sato theory, we construct the additional symmetry of the noncommutative  KP hierarchy and later the noncommutative  KP hierarchy with self-consistent sources  and the constrained noncommutative KP hierarchy will be constructed. In Section 4. We further do a reduction on the noncommutative  KP hierarchy to get the noncommutative  Gelfand-Dickey hierarchy and we also construct two new noncommutative systems including odd noncommutative C type Gelfand-Dickey and even noncommutative C type Gelfand-Dickey hierarchies. In Section 5,  we construct the additional symmetry of the noncommutative  Gelfand-Dickey hierarchy and meanwhile we define the noncommutative  Gelfand-Dickey hierarchy with self-consistent sources. The String equations of the noncommutative  Gelfand-Dickey hierarchy will be studied in Section 6.

\sectionnew{Noncommutative  KP hierarchy}
The KP  hierarchy
  is one of the most important topics  in the
area of classical integrable systems. In the noncommutative system, Moyal product $\star$ is defined by a skew symmetric matrix $\theta_{uv}$ as
\begin{equation*}
f\star g=\exp\Big(\frac{i}{2}\theta_{uv}{\d_{t_u}}{\d_{\bar t_v}}\Big)f(t)g(\bar t)\mid_{t=\bar t=t}
        =f(t)g(t)+\frac{i}{2}\theta_{uv}\d_{t_u}f(t)\d_{t_v}g(t)+\vartheta(\theta^2)
\end{equation*}
where $\vartheta(\theta^2)$ means the higher order terms. The matrix $\theta = (\theta_{uv})$
can contain functions in the variables $t_u$.
We can get that $[t_u,t_v]_{\star}=t_u\star t_v-t_v\star t_u=i\theta_{uv}$,
and when $\theta_{uv}\rightarrow0$, the noncommutative system can be reduced to the commutative ones.
The noncommutative KP  hierarchy is constructed by the pseudo-differential
operator $L=\d+u_2\d^{-1}+u_3\d^{-2}+....$ like this:
\[L_{t_n}=[B_n,L]_{\star}:= B_n\star L-L\star B_n,\]
 where $B_n=(L^n)_+$ and ``+" means the projection on nonnegative powers of
$\partial$.
In order to define the noncommutative   KP hierarchy, we need a formal adjoint operation $*$ for an arbitrary pseudo-differential operator $P=\sum\limits_{i}p_{i}\star\d^{i}$, we have $P^{*}=\sum\limits_{i}(-1)^{i}\d^{i}\star p_{i}$. Meanwhile, we have $\d^{*}=-\d$, $(\d^{-1})^{*}=-\d^{-1}$, and $(A\star B)^{*}=B^{*}\star A^{*}$ for two operators. The noncommutative  KP hierarchy in the Lax equation has the form
\begin{equation}
\frac{\d L}{\d t_{n}}=[B_{n},L]_{\star},\ \ n=1,2,...,
\end{equation}
where $u_{i}=u_{i}(t_{1}, t_{2}, t_{3}, ...)$.\\

From the first equation of the noncommutative
KP hierarchy we have that
$\partial_x = \partial_{t_1}$,
hence a solution depends only
on $x + t_1; t_2;\dots.$ The combination $x + t_1$ will be relabeled by $x$.
This hierarchy contains the $(2+1)$ dimensional noncommutative KP equation:
\begin{equation}\label{noncommutativeKPequation}
v_{t_3}-\frac14v_{xxx}-\frac34v_{x}\star v-\frac34v\star v_{x}-\frac34\int v_{t_2,t_2}dx+\frac34[v,\int v_{t_2}d x]_{\star}=0,
\end{equation}
 where $v=\frac12{u_2}$.

Meanwhile, we can give the noncommutative  KP hierarchy by the consistent conditions of the following set of linear partial differential equations
\begin{equation}
L\star\omega(t,\lambda)=\lambda\star\omega(t,\lambda),\ \ \frac{\d \omega(t,\lambda)}{\d t_{n}}=B_{n}\star\omega(t,\lambda),\ \ t=(t_{1}, t_{2}, ...).
\end{equation}
Here, $\omega(t,\lambda)$ is defined as a wave function, and let $\phi=1+\sum\limits_{i=1}^{\infty}\omega_{i}\d^{-i}$ be the wave operator of the noncommutative  KP hierarchy. The Lax operator and the wave function can be represented as
\begin{equation}\label{e}
L=\phi\star\d\star\phi^{-1},\ \ \omega(t,\lambda)=\phi(t)\star e^{\xi(t,\lambda)},
\end{equation}
in which $\xi(t,\lambda)=\lambda\star t_{1}+\lambda^{2}\star t_{2}+...$. \\
The Lax equation is equivalent to the following Sato equation
\begin{equation}\label{h}
\frac{\d \phi}{\d t_{n}}=-L_{-}^{n}\star\phi.
\end{equation}

\sectionnew{Additional symmetry of noncommutative   KP  hierarchy}
Firstly, we define the operator $\Gamma$ and the Orlov-Shulman's operator $M$ as
\begin{equation}\label{kpGDm}
M=\phi\star\Gamma\star \phi^{-1},\ \ \Gamma=\sum\limits_{i=0}^{\infty}it_{i}\d^{i}.
\end{equation}
Meanwhile, the Orlov-Shulman's operator $M$ satisfies the following identities
\begin{equation}
[L,M]_{\star}=1,\ \ \d_{t_{n}}M=[B_{n},M]_{\star},\ \ M\star\omega(t,z)=\d_{z}\omega(t,z).
\end{equation}
Further, we can get
\begin{equation}
\frac{\d M^{m}}{\d t_{n}}=[B_{n},M^{m}]_{\star},\ \ \frac{\d M^{m}L^{l}}{\d t_{n}}=[B_{n},M^{m}L^{l}]_{\star}.
\end{equation}
Moreover, to the wave function $\omega(t,z)$, $(L,M)$ is anti-isomorphic to $(z,\d_{z})$ with $[z,\d_{z}]_{\star}=-1$. We get
\begin{equation}
M^{m}\star L^{l}\star\omega(t,z)=z^{l}\star(\d_{z}^{m}\omega(t,z)),\ \  L^{l}\star M^{m}\star\omega(t,z)=\d_{z}^{m}(z^{l}\star\omega(t,z)).
\end{equation}
Next, we should consider the adjoint wave function $\omega^{*}$ and the adjoint Orlov-Shulman's operator $M^{*}$ that are useful for constructing the additional symmetry of the noncommutative  KP hierarchy, we have
\begin{equation}
\omega^{*}(t,z)=(\phi^{*})^{-1}\star e^{-\xi(t,z)},\ \ \xi(t,z)=\sum\limits_{i=0}^{\infty}t_{i}\lambda^{i}.
\end{equation}

However, $L^{*}$ and $M^{*}$ satisfy $[L^{*},M^{*}]_{\star}=[M,L]_{\star}^{*}=-1$.
Furthermore, we have
\begin{equation}
L^{*}\star\omega^{*}=z\star\omega^{*},  \d_{t_{n}}\omega^{*}=-B_{n}^{*}\star\omega^{*}.
\end{equation}
For $B_{n}^{*}\star\omega^{*}$, if the operator $A$ is a differential operator and has form $A:=\sum_{n=0}^{\infty}\partial^n a_n,$ then we define $A^{\ast} \star g(x)=\sum_{m=0}^{\infty}(-1)^{m}(\partial^{m}g(x))\star a_m.$
Next, we give the additional symmetries of the noncommutative  KP hierarchy. Firstly, we introduce additional independent variables $t_{m,l}^{*}$ and define the action of the additional flows on the wave operator as
\begin{equation}\label{kpGDf}
\frac{\d \phi}{\d t_{m,l}^{*}}=-(C_{m,l})_{-}\star \phi,
\end{equation}
in which \begin{equation}\label{kpGDl}
C_{m,l}=M^{m}\star L^{l},\ \ m,l\in \Z.
\end{equation}
In addition, we give some useful identities in the following proposition.
\begin{proposition}\label{kpGDg}The following identities hold true
\begin{equation}\label{LnMm}
\frac{\d L^{n}}{\d t_{m,l}^{*}}=-[(C_{m,l})_{-},L^{n}]_{\star}, \frac{\d M^{m}}{\d t_{m,l}^{*}}=-[(C_{m,l})_{-},M^{m}]_{\star},
\end{equation}
\begin{equation}\label{LnMm2}
\frac{\d C_{n,k}}{\d t_{m,l}^{*}}=-[(C_{m,l})_{-},C_{n,k}]_{\star}, \frac{\d C_{n,k}}{\d t_{n}}=[B_{n},C_{n,k}]_{\star}.
\end{equation}
\end{proposition}
\begin{proof}
The proof is very classical. We only need to consider the corresponding dressing structures
\begin{equation}
L^n=\phi\star\d^n\star\phi^{-1},\ \ M^m=\phi\star\Gamma^m\star\phi^{-1},
\end{equation}
and a direct calculation will lead to the results \eqref{LnMm} using the additional Sato equation\eqref{kpGDf}.
Similarly we can also consider the dressing structure
\begin{equation}
C_{n,k}=\phi\star \Gamma^{n}\star \partial^{k}\star\phi^{-1},\
\end{equation}
a direct calculation will lead to the results \eqref{LnMm2} using the additional Sato equation\eqref{kpGDf}.
We need to note that $\Gamma$ depends on $t_n$ and does not depend on $ t_{m,l}^{*}$.
\end{proof}

\begin{proposition}
The additional flows $\d_{t_{m,l}^{*}}$ commute with the flows $\d_{t_{n}}$ of the noncommutative  KP hierarchy, which can be shown as
\begin{equation}
[\d_{t_{m,l}^{*}},\d_{t_{n}}]_{\star}=0.
\end{equation}
\end{proposition}
\begin{proof}

Let the additional flows $\d_{t_{m,l}^{*}}$ and the  flows $\d_{t_{2n+1}}$ act on $\phi$,  we can get
\begin{eqnarray*}
[\d_{t_{m,l}^{*}},\d_{t_{n}}]_{\star}\phi&=&\d_{t_{m,l}^{*}}(\d_{t_{n}}\phi)-\d_{t_{n}}(\d_{t_{m,l}^{*}}\phi)\\
&=&-\d_{t_{m,l}^{*}}(L_{-}^{n}\star\phi)+\d_{t_{n}}((C_{m,l})_{-}\star\phi)\\
&=&-(\d_{t_{m,l}^{*}}L^{n})_{-}\star\phi-(L^{n})_{-}\star(\d_{t_{m,l}^{*}}\phi)\\
&&+(\d_{t_{n}}C_{m,l})_{-}\star\phi+(C_{m,l})_{-}\star(\d_{t_{n}}\phi)\\
&=&([(C_{m,l})_{-},L^{n}]_{\star})_{-}\star\phi+(L^{n})_{-}\star(C_{m,l})_{-}\star\phi\\
&&+([(L^{n})_{+},C_{m,l}]_{\star})_{-}\star\phi-(C_{m,l})_{-}\star(L^{n})_{-}\star\phi\\
&=&([(C_{m,l})_{-},L^{n}]_{\star})_{-}\star\phi-([(C_{m,l})_{-},L_{+}^{n}]_{\star})_{-}\star\phi+[L_{-}^{n},(C_{m,l})_{-}]_{\star}\star\phi\\
&=&([(C_{m,l})_{-},L_{-}^{n}]_{\star})_{-}\star\phi+[L_{-}^{n},(C_{m,l})_{-}]_{\star}\star\phi\\
&=&0.
\end{eqnarray*}
In the above proof,  $([L_{+}^{n},(C_{m,l})_{+}]_{\star})_{-}=0$ and  $([L_{+}^{n},C_{m,l}]_{\star})_{-}=([L_{+}^{n},(C_{m,l})_{-}]_{\star})_{-}$ have been used.
\end{proof}
Therefore, the additional flows $\d_{t_{m,l}^{*}}$ are symmetries of the noncommutative  KP hierarchy.
\begin{proposition}
The additional symmetry flows $\d_{t_{m,l}^{*}}$ of the noncommutative  KP hierarchy form a centerless $W_{1+\infty}$ algebra.
\end{proposition}
\begin{proof}. By a direct calculation, we can easily derive
\begin{equation}
[C_{m,l},C_{n,k}]_{\star}=\sum\limits_{p,q}\textbf{C}_{nk,ml}^{pq}C_{p,q},
\end{equation}
because the classical structure
\begin{equation}
[z^l\partial_m,z^k\partial_n]_{\star}=\sum\limits_{p,q}\textbf{C}_{nk,ml}^{pq}z^q\partial_p.
\end{equation}
 This further implies that
\begin{equation}\label{kpGDk}
([C_{m,l},C_{n,k}]_{\star})_{-}=-\sum\limits_{p,q}\textbf{C}_{nk,ml}^{pq}(C_{p,q})_{-},
\end{equation}
where the coefficient $\textbf{C}_{nk,ml}^{pq}$ is the standard coefficient of the W algebra \cite{C}.

Using eq.(\ref{kpGDk}), we can get
\begin{eqnarray*}
[\d_{t_{m,l}^{*}},\d_{t_{n,k}^{*}}]_{\star}\phi=-\sum\limits_{p,q}\textbf{C}_{nk,ml}^{pq}(C_{p,q})_{-}\star \phi=\sum\limits_{p,q}\textbf{C}_{nk,ml}^{pq}\d_{t_{p,q}^{*}}\phi,
\end{eqnarray*}
which is equal to
\begin{eqnarray*}
[\d_{t_{m,l}^{*}},\d_{t_{n,k}^{*}}]_{\star}=\sum\limits_{p,q}\textbf{C}_{nk,ml}^{pq}\d_{t_{p,q}^{*}}.
\end{eqnarray*}
\end{proof}

To have a better understanding of the additional symmetry flows in noncommutative case, we give a typical example of the noncommutative   KP hierarchy.
Let $m=1$, the corresponding flow on $L$ is
\begin{equation}\label{kpGDq}
\frac{\d L}{\d t_{1,1}^{*}}=-[(M\star L)_{-},L]_{\star}=L+[(M\star L)_{+},L]_{\star}.
\end{equation}
 Using eq.(\ref{kpGDm}), $ M\star L$ is expressed by
\begin{equation}\label{kpGDn}
 M\star L=\phi\star  x\d\star  \phi^{-1}+\sum\limits_{i}i\phi\star t_{i}\star \phi^{-1}\star \phi\d^{i}\star \phi^{-1}.
\end{equation}
Furthermore, using $\d x=x\d+1$ and $\d^{-i}x=x\d^{-i}-i\d^{-i-1}$, we get
\begin{equation}\label{kpGDo}
( \phi \star x\d\star  \phi^{-1})_{+}=x\d+\omega_{1}\star x-x\star\omega_{1},
\end{equation}
with the $ \phi^{-1}=1-\omega_{1}\d^{-1}+...$ being used. Taking eq.(\ref{kpGDo}) into eq.(\ref{kpGDn}), we have
\begin{equation}\label{kpGDp}
( M\star L)_{+}=x\d+\omega_{1}\star x-x\star\omega_{1}+\sum\limits_{i}i(\phi\star t_{i}\star \phi^{-1}\star L^{i})_{+}.
\end{equation}
Taking eq.(\ref{kpGDp}) into eq.(\ref{kpGDq}), we have
\begin{eqnarray*}
\frac{\d L}{\d t_{1,1}^{*}}&=&L+[x\d,L]_{\star}+[\sum\limits_{i}i(\phi\star t_{i}\star \phi^{-1}\star L^{i})_{+},L]_{\star}\\
&&+[\omega_{1}\star x,L]_{\star}-[x\star\omega_{1},L]_{\star}.
\end{eqnarray*}

Next, we define one generating function of the additional symmetries of the noncommutative  KP hierarchy. Firstly, we  define a generating operator $Y(\lambda,\mu)$ of the additional symmetries as
\begin{eqnarray*}
Y(\lambda,\mu)&=&\sum\limits_{m=0}^{\infty}\frac{(\mu-\lambda)^{m}}{m!}\sum\limits_{l=-\infty}^{\infty}\lambda^{-l-m-1}(C_{m,m+l})_{-},
\end{eqnarray*}
which can be expressed by a simple form in the following proposition.
\begin{proposition}The following identity holds true
\begin{equation}
Y(\lambda,\mu)=\frac{1}{\lambda}\omega(t,-\lambda)\d^{-1}\star\omega^*(t,\mu).
\end{equation}
\end{proposition}
\begin{proof}. Basing on eq.(\ref{kpGDf}), and the above lemma, we get
\begin{eqnarray*}
(M^{m}\star L^{m+l})_{-}&=&\sum\limits_{i=1}^{\infty}\d^{-i}res_{z}[z^{-1}\star(\d^{i-1}(M^{m}\star W\d^{m+l+1}\star e^{\xi}))\star(W^{-1}\star e^{-\xi})^{*}],\\
&=&\sum\limits_{i=1}^{\infty}res_{z}[z^{m+l}\d^{-i}(M^{m}\star\omega)\partial^{i-1}\star\omega^{*}(t,z)],\\
&=&res_{z}[z^{m+l}\star(\d_{z}^{m}\omega)\d^{-1}\star\omega^{*}(t,z)],
\end{eqnarray*}
with the help of the identity $f\d^{-1}=\d^{-1}f+\d^{-1}f_{x}\d^{-1}$.\\

Thus, we have one generating function of the noncommutative  KP hierarchy,
\begin{eqnarray*}
Y(\lambda,\mu)&=&\sum\limits_{m=0}^{\infty}\frac{(\mu-\lambda)^{m}}{m!}\star \sum\limits_{l=-\infty}^{\infty}\lambda^{-l-m-1}\star res_{z}[z^{l+m}\star(\d_{z}^{m}\omega(t,z))\d^{-1}\omega^{*}(t,z)]\\
&=&res_{z}[\sum\limits_{n=-\infty}^{+\infty}\frac{z^{n}}{\lambda^{n+1}}\star\omega(t,z+\mu-\lambda)\star\d^{-1}\omega^{*}(t,z)]\\
&=&\frac{1}{\lambda}[\omega(t,-\lambda)\d^{-1}\star\omega^{*}(t,\mu)].
\end{eqnarray*}
\end{proof}
\subsection{Noncommutative   KP  hierarchy with self-consistent sources}

The noncommutative   KP  hierarchy with self-consistent sources can be constructed by taking derivatives with respect to a new time variable $y_n$
\[L_{y_n}=[B_n+P(t)\d^{-1}\star Q(t),L]_{\star},\ \]
where
\[\d_{t_{n}}P=B_{n}\star P,\ \ \d_{t_{n}}Q=-B_{n}^{*}\star Q.\]
Here $P,Q$ take specific values of  $\omega(t,-\lambda),\omega^{*}(t,\mu)$ respectively.
And the corresponding Sato equation becomes
\[\phi_{y_n}=(-L^{n}_-+P(t)\d^{-1}\star Q(t))\star \phi.\]

\subsection{Constrained noncommutative  KP hierarchy}

The constrained noncommutative KP  hierarchy \cite{sole} is constructed by the following pseudo-differential Lax
operator $\bar L=\d+q(t)\d^{-1}\star r(t)$
\[\bar L_{t_n}=[\bar B_n,\bar L]_{\star},\ \ \bar B_n=\bar L^n_+,\]
where
\[\d_{t_{n}}q=\bar B_{n}\star q,\ \ \d_{t_{n}}r=-\bar B_{n}^{*}\star r.\]

 In order to get the explicit form of the flow equations, we need the operator $\bar B_{n}$,
\begin{eqnarray*}
\bar B_1&=&\partial,\\
\bar B_2&=&\d^2+2(q\star r),\\
\bar B_3&=&\d^3+3q\star r\star\d+3q_x\star r,\\
\dots\ \  &\dots& \ \ \ \dots.
\end{eqnarray*}
Then we can get  the first few flows of the
noncommutative cKP hierarchy
\begin{equation}\label{t1flow}
\begin{cases}
q_{t_1}=q_x\\
r_{t_1}=r_x,\\
\end{cases}
\end{equation}

\begin{equation}\label{t2flow}
\begin{cases}
q_{t_2}=q_{xx}+2(q\star r\star q)\\
r_{t_2}=-r_{xx}-2(r\star q\star r),\\
\end{cases}
\end{equation}

\begin{equation}\label{t3flow}
\begin{cases}
q_{t_3}&=q_{xxx}+3q_x\star r\star q+3q\star r\star q_x\\
&+q\star r_x\star q+r\star q_x\star q\\
r_{t_3}&=r_{xxx}+3r_x\star q\star r+3r\star q\star r_x,\\
\end{cases}
\end{equation}
\begin{equation*}
\dots\ \  \dots \ \ \ \dots.
\end{equation*}

\sectionnew{Noncommutative  Gelfand-Dickey hierarchies}
The  Gelfand-Dickey  hierarchy
  is one of the most important topics  in the
area of classical integrable systems.
The noncommutative Gelfand-Dickey  hierarchy can be constructed by the pseudo-differential
operator $\L=\d^N+u_2\d^{N-2}+u_3\d^{N-3}+...u_0$ like this:
\[\L_{t_n}=[\B_n,\L]_{\star}:= \B_n\star \L-\L\star \B_n,\ \ n\neq 0\  mod\  N,\]
 where $\B_n=(\L^{\frac nN})_+$ and ``+" means the  projection on nonnegative powers of $\partial$.

We can give the noncommutative  Gelfand-Dickey hierarchies by the consistent conditions of the following set of linear partial differential equations
\begin{equation}
\L\star\hat\omega(t,\lambda)=\lambda\star\hat\omega(t,\lambda),\ \ \frac{\d \hat\omega(t,\lambda)}{\d t_{n}}=\B_{n}\star\hat\omega(t,\lambda),\ \ t=(t_{1}, t_{2}, ...,t_{iN-1},t_{iN+1},...).
\end{equation}
Here, $W(t,\lambda)$ is defined as a wave operator, and let $W=1+\sum\limits_{i=1}^{\infty}\alpha_{i}\star\d^{-i}$ be the wave operator of the noncommutative  Gelfand-Dickey hierarchies. The Lax operator and the wave function can be represented as
\begin{equation}\label{e}
\L=W\star\d\star W^{-1},\ \ \hat\omega(t,\lambda)=W(t)\star e^{\xi_g(t,\lambda)},
\end{equation}
in which $\xi_g(t,\lambda)=\lambda\star t_{1}+\lambda^{2}\star t_{2}+...+\lambda^{iN-1}\star t_{iN-1}+\lambda^{iN+1}\star t_{iN+1}+...$. \\
The Lax equation is equivalent to the following Sato equation
\begin{equation}\label{h}
\frac{\d W}{\d t_{n}}=-\L_{-}^{\frac nN}\star W.
\end{equation}
When $N=2$, one can derive the noncommutative KdV hierarchy which contains the following noncommutative KdV equation
\[u_{t_3}=\frac14u_{xxx}+\frac34(u_x\star u+u\star u_x),\]
and the following noncommutative fifth order KdV equation
\[u_{t_5}=\frac1{16}u_{xxxxx}+\frac5{16}(u_{xxx}\star u+u\star u_{xxx})+\frac 58(u_x\star u_x+u\star u\star u)_x.\]
When $N=3$, one can derive the noncommutative Boussinesq hierarchy which contains the following noncommutative Boussinesq equation
\[u_{t_2,t_2}=\frac13u_{xxx}+(u\star u)_{xx}+([u,\d^{-1} u_{t_2}]_{\star})_x.\]

\subsection{Odd noncommutative C type Gelfand-Dickey hierarchies}

The Odd noncommutative C type Gelfand-Dickey  hierarchy can be constructed by the differential
operator
\[\notag\L=&&\d^{2N+1}+(\d^{N}u_1\d^{N-1}+\d^{N-1}u_1\d^{N})+(\d^{N-1}u_2\d^{N-2}+\d^{N-2}u_2\d^{N-1})\\
&&+...+(\d u_N +u_N\d),\]
which satisfies
\[\L^*=-\L.\]
The Lax equation of the Odd noncommutative C type  Gelfand-Dickey hierarchy is as
\[\L_{t_n}=[\B_n,\L]_{\star}:= \B_n\star \L-\L\star \B_n,\ \ n\neq 0\  mod\  2N+1,\]
 where $\B_n=(\L^{\frac n{2N+1}})_+$ and ``+" means the nonnegative projection on powers of $\partial$.
While, we can give the Odd noncommutative C type  Gelfand-Dickey hierarchies by the consistent conditions of the following set of linear partial differential equations
\begin{equation}
\L\star\omega_g(t,\lambda)=\lambda\star\omega_g(t,\lambda),\ \ \frac{\d \omega_g(t,\lambda)}{\d t_{n}}=\B_{n}\star\omega_g(t,\lambda),
\end{equation}
where
$t=(t_{1}, t_{2}, ...,t_{i},...),\ i\neq 0\  mod\  2N+1.$
Here, $W(t,\lambda)$ is defined as a wave function, and let $W=1+\sum\limits_{i=1}^{\infty}\alpha_{i}\star\d^{-i}$ be the wave operator of the noncommutative  Gelfand-Dickey hierarchies. The Lax operator and the wave function can be represented as
\begin{equation}\label{e}
\L=W\star\d\star W^{-1},\ \ \omega_g(t,\lambda)=W(t)\star e^{\xi_g(t,\lambda)},
\end{equation}
in which $\xi_g(t,\lambda)=\lambda\star t_{1}+\lambda^{2}\star t_{2}+...+\lambda^{i(2N+1)-1}\star t_{i(2N+1)-1}+\lambda^{i(2N+1)+1}\star t_{i(2N+1)+1}+...$. \\
The Lax equation is equivalent to the following Sato equation
\begin{equation}\label{h}
\frac{\d W}{\d t_{n}}=-\L_{-}^{\frac nN}\star W.
\end{equation}

\subsection{Even noncommutative C type Gelfand-Dickey hierarchies}

The noncommutative Gelfand-Dickey  hierarchy can be constructed by the differential
operator
\[\notag\L=\d^{2N}+\d^{N-1}u_1\d^{N-1}+...+\d u_{N-1}\d+u_N,\]
which satisfies
\[\L^*=\L.\]
The Lax equation of the Odd noncommutative C type  Gelfand-Dickey hierarchy is as
\[\L_{t_n}=[\B_n,\L]_{\star}:= \B_n\star \L-\L\star \B_n,\ \ n\neq 0\  mod\  2N,\]
 where $\B_n=(\L^{\frac n{2N}})_+$ and ``+" means the nonnegative projection on powers of $\partial$.

\sectionnew{Additional symmetry of noncommutative  Gelfand-Dickey hierarchies}
Firstly, we define the operator $\Gamma_g$ and the Orlov-Shulman's operator $\M$ as
\begin{equation}\label{GDm}
\M=W\star\Gamma_g\star W^{-1},\ \ \Gamma_g=\sum\limits_{j=0}^{N-1}\sum\limits_{i=0}^{\infty}(Ni+j)t_{Ni+j}\d^{Ni+j}.
\end{equation}
Meanwhile, the Orlov-Shulman's operator $\M$ satisfy the following identities
\begin{equation}
[\L,\M]_{\star}=1,\ \ \d_{t_{n}}\M=[\B_{n},\M]_{\star},\ \ \M\star\hat\omega(t,z)=\d_{z}\hat\omega(t,z).
\end{equation}
Further, we can get
\begin{equation}
\frac{\d \M^{m}}{\d t_{n}}=[\B_{n},\M^{m}]_{\star},\ \ \frac{\d \M^{m}\L^{l}}{\d t_{n}}=[\B_{n},\M^{m}\L^{l}]_{\star}.
\end{equation}
Moreover, to the wave function $\hat\omega(t,z)$, $(\L,\M)$ is anti-isomorphic to $(z,\d_{z})$ with $[z,\d_{z}]_{\star}=-1$. We get
\begin{equation}
\M^{m}\star \L^{l}\star\hat\omega(t,z)=z^{l}\star(\d_{z}^{m}\hat\omega(t,z)),\ \  \L^{l}\star \M^{m}\star\hat\omega(t,z)=\d_{z}^{m}(z^{l}\star\hat\omega(t,z)).
\end{equation}
Next, we should consider the adjoint wave function $\hat\omega^{*}$ and the adjoint Orlov-Shulman's operator $\M^{*}$ that are useful for constructing the additional symmetry of the noncommutative Gelfand-Dickey hierarchies, we have
\begin{equation}
\hat\omega^{*}(t,z)=(W^{*})^{-1}\star e^{-\xi_g(t,z)}.
\end{equation}

However, $\L^{*}$ and $\M^{*}$ satisfy $[\L^{*},\M^{*}]_{\star}=[\M,\L]_{\star}^{*}=-1$.
Furthermore, we have
\begin{equation}
\L^{*}\star\hat\omega^{*}=z\star\hat \omega^{*},  \d_{t_{n}}\hat\omega^{*}=-\B_{n}^{*}\star\hat\omega^{*}.
\end{equation}
Next, we give the additional symmetries of the noncommutative Gelfand-Dickey hierarchies. Firstly, we introduce additional independent variables $t_{m,l}^{*}$ and define the action of the additional flows on the wave operator as
\begin{equation}\label{GDf}
\frac{\d W}{\d t_{m,l}^{*}}=-(D_{m,l})_{-}\star W,
\end{equation}
in which \begin{equation}\label{GDl}
D_{m,l}=\M^{m}\star \L^{l}.
\end{equation}
\begin{remark}
For the C type noncommutative Gelfand-Dickey hierarchy, to keep the C type condition the operator $D_{m,l}$ needs to take the following form
\begin{equation}\label{CGDl}
D_{m,l}=\M^{m}\star \L^{l}-(-1)^{l}\L^{l}\star \M^{m}.
\end{equation}
\end{remark}
In addition, we give some useful identities in the following proposition.
\begin{proposition}\label{GDi}The following identities hold true
\begin{equation}
\frac{\d \L^{n}}{\d t_{m,l}^{*}}=-[(D_{m,l})_{-},\L^{n}]_{\star}, \frac{\d \M^{m}}{\d t_{m,l}^{*}}=-[(D_{m,l})_{-},\M^{m}]_{\star},
\end{equation}
\begin{equation}
\frac{\d D_{n,k}}{\d t_{m,l}^{*}}=-[(D_{m,l})_{-},D_{n,k}]_{\star}, \frac{\d D_{n,k}}{\d t_{n}}=[\B_{n},D_{n,k}]_{\star}.
\end{equation}
\end{proposition}
\begin{proof}
The proof is very classical. We only need to consider the corresponding dressing structures
\begin{equation}
\L^n= W \star\d^n\star W ^{-1},\ \ \M^m= W \star\Gamma_g^m\star W ^{-1},
\end{equation}
and
\begin{equation}
D_{n,k}= W \star (\Gamma_g^{n}\star \partial^{k}-(-1)^{k}\partial^{k}\star\Gamma_g^{n})\star W ^{-1},\
\end{equation}
a direct calculation will lead to the proposition using the additional Sato equation.
We need to note that $\Gamma_g$ depends on $t_n$ and does not depend on $ t_{m,l}^{*}$.
\end{proof}

\begin{proposition}
The additional flows $\d_{t_{m,l}^{*}}$ commute with the flows $\d_{t_{n}}$ of the noncommutative Gelfand-Dickey hierarchies, which can be shown as
\begin{equation}
[\d_{t_{m,l}^{*}},\d_{t_{n}}]_{\star}=0,\ n\neq \ 0\ mod\ N.
\end{equation}
\end{proposition}
\begin{proof}

Let the additional flows $\d_{t_{m,l}^{*}}$ and the  flows $\d_{t_{n}}$ act on $ W $,  we can get
\begin{eqnarray*}
[\d_{t_{m,l}^{*}},\d_{t_{n}}]_{\star} W &=&\d_{t_{m,l}^{*}}(\d_{t_{n}} W )-\d_{t_{n}}(\d_{t_{m,l}^{*}} W )\\
&=&-\d_{t_{m,l}^{*}}(\L_{-}^{\frac nN}\star W )+\d_{t_{n}}((D_{m,l})_{-}\star W )\\
&=&-(\d_{t_{m,l}^{*}}\L^{\frac nN})_{-}\star W -(\L^{\frac nN})_{-}\star(\d_{t_{m,l}^{*}} W )\\
&&+(\d_{t_{n}}D_{m,l})_{-}\star W +(D_{m,l})_{-}\star(\d_{t_{n}} W )\\
&=&([(D_{m,l})_{-},\L^{\frac nN}]_{\star})_{-}\star W +(\L^{\frac nN})_{-}\star(D_{m,l})_{-}\star W \\
&&+([(\L^{\frac nN})_{+},D_{m,l}]_{\star})_{-}\star W -(D_{m,l})_{-}\star(\L^{\frac nN})_{-}\star W \\
&=&([(D_{m,l})_{-},\L^{\frac nN}]_{\star})_{-}\star W -([(D_{m,l})_{-},\L_{+}^{\frac nN}]_{\star})_{-}\star W +[\L_{-}^{\frac nN},(D_{m,l})_{-}]_{\star}\star W \\
&=&([(D_{m,l})_{-},\L_{-}^{\frac nN}]_{\star})_{-}\star W +[\L_{-}^{\frac nN},(D_{m,l})_{-}]_{\star}\star W \\
&=&0.
\end{eqnarray*}
In the above proof,  $([\L_{+}^{\frac nN},(D_{m,l})_{+}]_{\star})_{-}=0$ and  $([\L_{+}^{\frac nN},D_{m,l}]_{\star})_{-}=([\L_{+}^{\frac nN},(D_{m,l})_{-}]_{\star})_{-}$ have been used.
\end{proof}

Therefore, the additional flows $\d_{t_{m,l}^{*}}$ are symmetries of the noncommutative Gelfand-Dickey hierarchies.
\begin{proposition}
The additional symmetry flows $\d_{t_{m,l}^{*}}$ of the noncommutative Gelfand-Dickey hierarchies form the centerless $W_{1+\infty}$.
\end{proposition}
\begin{proof}. By a direct calculation, we can easily derive
\begin{equation}
[D_{m,l},D_{n,k}]_{\star}=\sum\limits_{p,q}\textbf{C}_{nk,ml}^{pq}D_{p,q},
\end{equation}
and it implies that
\begin{equation}\label{GDk}
([D_{m,l},D_{n,k}]_{\star})_{-}=-\sum\limits_{p,q}\textbf{C}_{nk,ml}^{pq}(D_{p,q})_{-},
\end{equation}
where the coefficient $\textbf{C}_{nk,ml}^{pq}$ is the standard coefficient of the W algebra \cite{C}.

Using eq.(\ref{GDk}), we can get
\begin{eqnarray*}
[\d_{t_{m,l}^{*}},\d_{t_{n,k}^{*}}]_{\star}W=-\sum\limits_{p,q}\textbf{C}_{nk,ml}^{pq}(C_{p,q})_{-}\star W=\sum\limits_{p,q}\textbf{C}_{nk,ml}^{pq}\d_{t_{p,q}^{*}}W,
\end{eqnarray*}
which is equal to
\begin{eqnarray*}
[\d_{t_{m,l}^{*}},\d_{t_{n,k}^{*}}]_{\star}=\sum\limits_{p,q}\textbf{C}_{nk,ml}^{pq}\d_{t_{p,q}^{*}}.
\end{eqnarray*}
\end{proof}

To have a better understanding of the additional symmetry flows in noncommutative case, we give a typical example of the noncommutative  Gelfand-Dickey hierarchies.
Let $m=1$, the corresponding flow on $\L$ is
\begin{equation}\label{GDq}
\frac{\d \L}{\d t_{1,1}^{*}}=-[(\M\star \L)_{-},\L]_{\star}=\L+[(\M\star \L)_{+},\L]_{\star}.
\end{equation}
 Using eq.(\ref{GDm}), $ \M\star \L$ is expressed by
\begin{equation}\label{GDn}
 \M\star \L=W x\d W^{-1}+\sum\limits_{i\neq 0\ mod \ N}iW\star t_{i}\star  W^{-1}\star W\star \d^{i}\star  W^{-1}.
\end{equation}
Furthermore, using $\d x=x\d+1$ and $\d^{-i}x=x\d^{-i}-i\d^{-i-1}$, we get
\begin{equation}\label{GDo}
( W \star x\star \d W^{-1})_{+}=x\d+\hat\omega_{1}\star x-x\star\hat\omega_{1},
\end{equation}
with the $ W^{-1}=1-\hat\omega_{1}\d^{-1}+...$ being used. Taking eq.(\ref{GDo}) into eq.(\ref{GDn}), we have
\begin{equation}\label{GDp}
( \M\star \L)_{+}=x\d+\hat\omega_{1}\star x-x\star\hat\omega_{1}+\sum\limits_{i\neq 0\  mod \ N}i(W\star t_{i}\star  W^{-1}\star \L^{i})_{+}.
\end{equation}
Taking eq.(\ref{GDp}) into eq.(\ref{GDq}), we have
\begin{eqnarray*}
\frac{\d \L}{\d t_{1,1}^{*}}&=&\L+[x\d,\L]_{\star}+[\sum\limits_{i\neq 0\  mod \ N}i(W\star t_{i}\star  W^{-1}\star \L^{i})_{+},\L]_{\star}\\
&&+[\hat\omega_{1}\star x,\L]_{\star}-[x\star\hat\omega_{1},\L]_{\star}.
\end{eqnarray*}

Next, we define one generating function of the additional symmetries of the noncommutative Gelfand-Dickey hierarchies.
We define a generating operator $Y_g(\lambda,\mu)$ of the additional symmetries as
\begin{eqnarray*}
Y_g(\lambda,\mu)&=&\sum\limits_{m=0}^{\infty}\frac{(\mu-\lambda)^{m}}{m!}\sum\limits_{l=-\infty}^{\infty}\lambda^{-l-m-1}(D_{m,m+l})_{-}\\
&=&\frac{1}{\lambda}\hat\omega(t,-\lambda)\d^{-1}\star\hat\omega^*(t,\mu).
\end{eqnarray*}

\subsection{Noncommutative  Gelfand-Dickey hierarchy with self-consistent sources}

The noncommutative  Gelfand-Dickey hierarchy with self-consistent sources can be constructed by the differential
operator $\L=\d^N+u_2\d^{N-2}+u_3\d^{N-3}+...u_0$ like this:
\[\L_{y_n}=[\B_n+\Phi(t)\d^{-1}\star \Psi(t),\L]_{\star},\ \ n\neq \ 0\ mod\ N,\]
where
\[\d_{t_{n}}\Phi=\B_{n}\star \Phi,\ \ \d_{t_{n}}\Psi=-\B_{n}^{*}\star \Psi.\]
And the corresponding Sato equation becomes
\[W_{y_n}=(-\L^{\frac nN}_-+\Phi(t)\d^{-1}\star \Psi(t))\star W.\]

\sectionnew{String equations of the noncommutative  Gelfand-Dickey hierarchy}
In this section, we will consider the string equation of the noncommutative  Gelfand-Dickey hierarchy. Firstly, we get a special action of the additional flows on $\L^{l}$
\begin{eqnarray}\label{GDs}
\notag \d_{t_{1,-(l-1)}^{*}}\L^l &=&[-(D_{1,-(l-1)})_{-},\L^{l}]_{\star} \\
 \notag &=&[(D_{1,-(l-1)})_{+},\L^{l}]_{\star}+[-D_{1,-(l-1)},\L^{l}]_{\star} \\
&=&[(\M\star \L^{-(l-1)})_{+},\L^{l}]_{\star}+l.
\end{eqnarray}
Basing on the above knowledge, we can get the following proposition on the String equation.
\begin{proposition}
If $\L^{l}$ is independent on the additional variable $t_{1,-(l-1)}^{*}$, then
\begin{equation}\label{GDt}
[\L^{l},\frac{1}{l}(\M\star \L^{-(l-1)})_{+}]_{\star}=1,\ l\in \Z \setminus \{0\},
\end{equation}
is a string equation similar as the standard form $[P,Q]_{\star}=1$ ($P$ and $Q$ are two differential operators) of the noncommutative  Gelfand-Dickey hierarchy.
\end{proposition}

{\bf {Acknowledgements:}}
  This work is supported by  National Natural Science Foundation of China under Grant No. 11571192,   and K. C. Wong Magna Fund in
Ningbo University.

\end{document}